\documentclass[a4paper,english,numberwithinsect,cleveref]{lipics-v2021}

\hideLIPIcs

\usepackage[utf8]{inputenc} % LaTeX source encoded as UTF-8
\title{Efficient attack sequences in m-eternal domination}

\author{V{\'a}clav Bla{\v z}ej}{Faculty of Information Technology, Czech Technical University in Prague, Prague, Czech Republic}{}{https://orcid.org/0000-0001-9165-6280}{}
\author{Jan Maty{\'a}{\v s} K{\v r}i{\v s}{\v t}an}{Faculty of Information Technology, Czech Technical University in Prague, Prague, Czech Republic}{}{https://orcid.org/0000-0001-6657-0020}{}
\author{Tom{\'a}{\v s} Valla}{Faculty of Information Technology, Czech Technical University in Prague, Prague, Czech Republic}{}{https://orcid.org/0000-0003-1228-7160}{}

\authorrunning{V. Bla{\v z}ej, J. M. K{\v r}i{\v s}{\v t}an, and T. Valla}
\Copyright{V{\' a}clav Bla{\v z}ej, Jan M. K{\v r}i{\v s}{\v t}an, and Tom{\' a}{\v s} Valla}
\funding{OP VVV MEYS funded project CZ.02.1.01/0.0/0.0/16\_019/0000765 ``Research Center for Informatics''. This work was supported by the Grant Agency of the Czech Technical University in Prague, grant \mbox{No.~SGS20/208/OHK3/3T/18}. }

\usepackage[nocompress]{cite}
\usepackage{lineno} %line numbers
\usepackage{todonotes}
\nolinenumbers%

\usepackage{amssymb} %additional math symbols
\usepackage{mathtools}
\usepackage{algorithm}
\usepackage{xspace}
\usepackage{algpseudocode}
\usepackage{thmtools} % theorem
\usepackage{hyperref}
\usepackage{url}
\usepackage{amsmath} %advanced maths
\usepackage{thm-restate}

\newtheorem{reduction}{Reduction}{}
\crefname{reduction}{reduction}{reductions}

\newcommand\putabove[2]{\mathrel{\overset{\makebox[0pt]{\mbox{\normalfont\tiny\sffamily #2}}}{#1}}}

\graphicspath{{./images/}}

\makeatletter
\newcommand{\customlabel}[2]{%
\protected@write \@auxout {}{\string \newlabel {#1}{{#2}{\thepage}{#2}{#1}{}} }%
\hypertarget{#1}{#2}
}
\makeatother

\newcommand{\EDN}{\gamma^{\infty}_\mathrm{m}} % one vertex <= 1 guard
\newcommand{\MINAT}{\beta} % domination number
\newcommand\DF{\delta}
\newcommand\DFT{\delta_T}
\newcommand\G{g}
\newcommand\GT{g_T}
\newcommand\CHLD{\text{chld}}
\newcommand\CN{\kappa}
\newcommand\PAR{p}
\newcommand\NEOCOL{\mathcal{V}}

\begin{document}

\maketitle

\begin{abstract}
  We study the m-eternal domination problem from the perspective of the attacker.
  For many graph classes, the minimum required number of guards to defend eternally is known.
  By definition, if the defender has less than the required number of guards, then there exists a sequence of attacks that ensures the attacker's victory.
  Little is known about such sequences of attacks, in particular, no bound on its length is known.

  We show that if the game is played on a tree $T$ on $n$ vertices and the defender has less than the necessary number of guards, then the attacker can win in at most $n$ turns.
Furthermore, we present an efficient procedure that produces such an attacking strategy.
\end{abstract}

%Research highlights
%\begin{highlights}
%\item Proof that the attacker can win in $\mathcal{O}(n)$ turns on trees in m-eternal domination
  %\item An efficient procedure that generates a sequence of attacks of length at most $\mathcal{O}(n)$, that ensures the attacker's.
%\end{highlights}

\section{Introduction}

Consider the following game, played by an attacker and a defender on graph $G$.
The defender controls a set of guards, which he initially places on the vertices of $G$.
Each vertex can be occupied by at most one guard.

In each turn, the attacker first chooses one vertex, which he \emph{attacks}.
The defender then must \emph{defend} against the attack by moving some or all of his guards along their adjacent edges, so that one of the guards moves to the attacked vertex.

If the attacked vertex is not occupied by a guard after the attack, the attacker wins.
The defender wins if he can defend indefinitely.
An \emph{m-eternal dominating set} is a set of vertices, which, when used as the starting configuration of guards, is winning for the defender.

Much previous research has focused on determining the minimum size of such a set, denoted by $\EDN(G)$.
It has been explored in several recent papers, see for instance \cite{henning2017bounds, eternal-dom-sets, eternal-security-in-graphs, proper-interval-graphs, blavzej2019m, Messinger2017}.
Some recent works have focused on trees in particular, such as \cite{klostermeyer2021eternal, henning2016trees}.

It is unknown if it is possible to decide in PSPACE whether a set of guards induces an m-eternal dominating set.
Klostermeyer et al.~\cite{survey-article} pose the following open problem: is there a function $\MINAT(n)$, where $n$ is the number of vertices of the graph $G$, such that if one can defend against $\MINAT(n)$ attacks with a given configuration of guards, it is possible to defend against any sequence of attacks?
No such $\MINAT(n)$ is currently known in the general case.
A polynomial upper bound on $\MINAT(n)$ would imply that the problem of deciding whether a set of vertices induces an m-eternal dominating set is in PSPACE, as it would suffice to try every possible response to every possible sequence of $\MINAT(n)$ attacks.

In this paper, we provide a linear bound on $\MINAT(n)$ when the input graph is a tree.  
The result follows from the following theorem.

\begin{restatable*}{theorem}{attackstrategy}
\label{thm:v_trees}
  Let $C$ be a configuration of guards on tree $T$ of at most $\EDN(T) - 1$ guards.
  Then the attacker can win against $C$ in at most ${\rm diam}(T)$ steps, where ${\rm diam}(T)$ is the diameter of $T$.
\end{restatable*}

%\subsection{Preliminaries}
All graphs considered are undirected and simple.
Let $T$ be a tree rooted in $r \in V(T)$.
Let $v \in V(T)$, then by $\CHLD(v)$ we denote the children of $v$ and by $\PAR(v)$ the parent of $v$.
By $T(v)$ we denote the subtree of $T$ rooted in $v$.
We say that a tree with no vertices of degree $2$ is a \emph{shrubbery}.
Let $G = (V, E)$ be a graph and $S \subseteq V$, then by $G[S]$ we denote the subgraph of $G$ induced by $S$.
By $K_n$ we denote the complete graph on $n$ vertices.
We say that a \emph{leaf} of a rooted tree is any vertex of degree 1, including the root.
We say that $v \in V(T)$ is an \emph{inner} vertex if v has degree at least 2.

\section{Efficient attack sequences on trees}

In this section, we present an explicit defending strategy on trees, whose definition will help us identify parts of tree $T$ which are vulnerable to attacks.
This is followed by a description of the attacking process, which guarantees the attacker a win as long as the number of guards is less than $\EDN(T)$.

%\begin{lemma}
%For a tree $T$ it holds $2 \geq \MINAT(G) \geq \MINAT{n/2}$ and equality holds for infinite classes of trees.
%\end{lemma}
%\begin{proof}
%The first equality holds for stars, the second one for odd paths.
%\end{proof}

\subsection{Defending strategy based on neo-colonization}

Let $T$ be the tree on which the game is played.
To describe the attacking strategy, we make use of the neo-colonization of $T$.
Let $\gamma_c(G)$ be the size of the minimum connected dominating set of a graph $G$.
\begin{definition}[Goddard et al.~\cite{eternal-security-in-graphs}]
  A \emph{neo-colonization} of a graph $G$ is a partition $\NEOCOL = \{V_1, \dots, V_t\}$ of the vertex set of $G$ such that each $G[V_i]$ is connected.
  Each part $V_i$ is assigned a weight $\omega(V_i)$ as follows.
  \[
    \omega(V_i) = \begin{cases}
      1 & \text{ if  $G[V_i]$ is a clique }\\
      \gamma_c(G[V_i]) + 1 & \text{ otherwise}
    \end{cases}
  \]
\end{definition}
By $\omega(\NEOCOL)$ we denote the total weight of the neo-colonization.
By $\theta_c(G)$ we denote the minimum total weight of any neo-colonization of $G$, and it is called the \emph{clique-connected cover number} of $G$.
Goddard et al.~\cite{eternal-security-in-graphs} proved that $\EDN(G) \leq \theta_c(G)$.

We define a special case of neo-colonization which will be useful in describing defending strategies on trees.

\begin{definition}\label{def:nice-neocol}
We say that a neo-colonization $\NEOCOL = \{V_1, \dots, V_t\}$ of a tree $T$ rooted at a leaf $r$ is \emph{nice} if for every $V_i \in \NEOCOL$
\begin{itemize}
  \item $T[V_i]$ is a shrubbery,
  \item the vertex with the minimum distance to $r$ in $V_i$ is a leaf in $T[V_i]$.
\end{itemize}
\end{definition}
By $\NEOCOL(v)$ we denote the part $V_i$ such that $v \in V_i$.

Klostermeyer et al.~\cite{eternal-dom-sets} proved that $\theta_c(T) = \EDN(T)$ on trees.
We will show that this equality holds even when we consider the minimum total weight among only nice neo-colonizations.
To that end, we first recall the linear algorithm that computes $\EDN$ in trees by Klostermeyer et al.~\cite{eternal-dom-sets}.
The algorithm is based on two reductions.

\begin{reduction}\label{rdc:leaves-orig}
Let $x$ be a vertex of $T$ incident to $\ell \geq 2$ leaves and to exactly one
vertex of degree at least two.
Delete all leaves incident to $x$.
\end{reduction}

\begin{reduction}\label{rdc:leaf}
Let $x$ be a vertex of degree two in $T$ such that $x$ is adjacent to exactly one leaf, say $y$.
Delete both $x$ and $y$.
\end{reduction}

\begin{lemma}[Klostermeyer et al.~{\cite[Lemma 20 and Lemma 21]{eternal-dom-sets}}]\label{lem:rdc-correct}
If $T'$ is the result of applying \Cref{rdc:leaves-orig} or \Cref{rdc:leaf} to the tree $T$, then $T'$ is a tree and $\EDN(T) = 1 + \EDN(T')$.
\end{lemma}
We derive the following reduction from \Cref{rdc:leaves-orig} so the subsequent analysis is made simpler.

\begin{reduction}\label{rdc:leaves}
Let $x$ be a vertex of $T$ such that all of its $\ell \geq 2$ children are leaves and it has a parent.
Delete all children of $x$.
\end{reduction}
Note that the proof of Lemma 20~\cite{eternal-dom-sets} is applicable to \Cref{rdc:leaves} as well, if we consider $w$, as denoted in the proof of Lemma 20 in~\cite{eternal-dom-sets}, to be the parent of $x$.
The argument only uses the fact that a vertex neighboring $x$ remains after the reduction, with no argument made based on its degree.
This implies the following lemma.
\begin{lemma}\label{lem:leaves-correct}
If $T'$ is the result of applying \Cref{rdc:leaves} to the tree $T$, then $T'$ is a tree and $\EDN(T) = 1 + \EDN(T')$.
\end{lemma}

\begin{lemma}\label{lem:reductions-result}
  Let $T$ be a tree rooted in a leaf $r$.
  After an exhaustive application of \Cref{rdc:leaf} or \Cref{rdc:leaves}, we are left with a $K_1$ or $K_2$.
\end{lemma}
\begin{proof}
  Suppose that $n \geq 3$ and $c \neq r$ is a leaf of $T$.
  If $\PAR(c)$ has no other children than $c$, then \Cref{rdc:leaf} is applicable.
  Otherwise, $\PAR(c)$ has more than $2$ children.
  Note that $\PAR(c) \neq r$ as $r$ was chosen to be a leaf of $T$, therefore $p$ has a parent as well and \Cref{rdc:leaves} is applicable.
  Thus, if neither reduction is applicable, then $n \leq 2$.
\end{proof}
The following lemma shows how those reductions can be used to construct a nice neo-colonization of any tree.

\begin{lemma}\label{lem:neocol-construction}
  For every tree $T$ there exists a nice neo-colonization $\NEOCOL$ with $\omega(\NEOCOL) = \EDN(T)$.
\end{lemma}
\begin{proof}
First, we root the tree $T$ in an arbitrary leaf $r$.
  In each step, we consider a leaf of maximum depth and apply either \Cref{rdc:leaf} or \Cref{rdc:leaves}, depending on which one is applicable.

  We construct an auxiliary graph $H$ on the vertices of $T$.
  Its connected components will induce the parts of the resulting nice neo-colonization $\NEOCOL$ of $T$.
  If we apply \Cref{rdc:leaves} with leaves $c_1, \dots, c_k$ and their parent $x$, then for each $i \in \{1, \dots, k\}$ we add $\{c_i, x\}$ to $E(H)$.
  If we apply \Cref{rdc:leaf} with leaf $y$ and its parent $x$, then we add $\{x, y\}$ to $E(H)$.

  By \Cref{lem:reductions-result} we eventually reduce $T$ to $K_1$ or $K_2$.
  If $T$ is reduced to $K_2$, we connect the two remaining vertices by an edge in $H$.
  \begin{claim}
    Let $V_1, \dots, V_t$ be the connected components of $H$.
    Then $\NEOCOL = \{V_1, \dots, V_t\}$ is a nice neo-colonization of $T$.
  \end{claim}
  \begin{proof}
  We show that $\NEOCOL$ satisfies the conditions of \Cref{def:nice-neocol}.
  Note that all $V_i, V_j \in \NEOCOL$ such that $V_i \neq V_j$ are vertex disjoint.
  Also note that in each $H[V_i]$ except $H[\NEOCOL(r)]$, there is exactly one edge, say $e = \{x, y\}$, constructed by the application of \Cref{rdc:leaf} and it is the last edge added to $H[V_i]$.
  %After applying \Cref{rdc:leaf}, it will have only one child in $H[V_i]$ and its parent will be outside of $H[V_i]$.
  Suppose that when applying \Cref{rdc:leaf}, $x$ was the vertex of degree two and $y$ was the leaf.
  While applying the reduction, $x$ is deleted, and therefore no other vertex than $y$ will be adjacent to $x$ in $H$.
  Thus $x$ is the vertex closest to $r$ in $V_i$ and its only neighbor in $H$ is $y$.

  We also show that $T[V_i]$ contains no vertex of degree $2$ for every $V_i \in \NEOCOL$.
  Let $h$ be the vertex of minimum depth in $V_i$.
  Either $V_i$ consists of $2$ vertices and both have degree $1$ in $T[V_i]$, or every vertex except $h$ was part of some application of \Cref{rdc:leaves}.
  This follows from the fact that for every $T[V_i]$ at most one edge was constructed after an application of \Cref{rdc:leaf}.
  Each $u \in V_i \setminus \{h\}$ thus appeared only as one of the leaves when applying \Cref{rdc:leaves} and has degree $1$ in $T[V_i]$, or at some point appeared as $x$ when applying \Cref{rdc:leaves}.
  In that case, $T[V_i]$ also contains at least two children of $u$.
  Furthermore, $x$ also appeared as a leaf in a subsequent \Cref{rdc:leaves} or as $y$ in a subsequent \Cref{rdc:leaf}.
  In both cases $T[V_i]$ also contains the parent of $u$, thus $u$ has degree at least $3$ in $T[V_i]$.
  \end{proof}

  We show that $\omega(\NEOCOL) = \EDN(T)$.
  Let $\rho$ be the total number of applications of Reductions~\ref{rdc:leaf} and \ref{rdc:leaves}.
  Let $T'$ be the resulting tree after exhaustively applying Reductions~\ref{rdc:leaf} and \ref{rdc:leaves}.
  From \Cref{lem:rdc-correct}, \Cref{lem:leaves-correct}, and \Cref{lem:reductions-result}, it follows that
  \begin{equation}\label{eq:algo-works}
    \EDN(T) = \rho + \EDN(T') = \rho + 1.
  \end{equation}
  Now, let us consider the weight of the individual colonies.
  If $H[V_i]$ consists of a single edge, then $\omega(V_i) = 1$.
  We show that otherwise $\omega(V_i) = \gamma_c(T[V_i]) + 1$.
  
  For any tree $T$ on $n$ vertices, the minimum connected dominating set of $T$ consists of the complement of the set of leaves~\cite{connected-domination-trees}.
  Suppose that $T[V_i]$ is not a clique, then $\omega(V_i) = d + 1$ where $d$ is the number of vertices of degree at least $3$ in $T[V_i]$.
  Note that every inner vertex of $T[V_i]$ was part of exactly one \Cref{rdc:leaves} as the parent vertex and also exactly one edge of $T[V_i]$ was added by an application of \Cref{rdc:leaf}.
  Let $\rho_i$ be the number of reductions which added an edge to $H[V_i]$.
  Let $d_i$ be the number of inner vertices of $T[V_i]$.
  It follows that 
  \begin{equation}\label{eq:part-weight}
    \omega(V_i) = d_i + 1 = \rho_i
  \end{equation}
  for all $V_i$ which had all of its edges in $H$ added by the reductions.

  It remains to check $\NEOCOL(r)$, which was created or modified when processing $T'$ at the end of the process.
  Without loss of generality, let $V_1 = \NEOCOL(r)$.
  If $T'$ is isomorphic to $K_1$, then $r$ is not adjacent to any other vertex in $H$.
  In that case $\rho_1 = 0$ and $\omega(\NEOCOL(r)) = \EDN(T') = 1$.

  Otherwise, $T'$ is isomorphic to $K_2$.
  Note that every inner vertex of $T[\NEOCOL(r)]$ is still a part of exactly one \Cref{rdc:leaves}.
  Let $u$ be the child of $r$ and let $d_r$ be the number of inner vertices in $T[\NEOCOL(r)]$.
  If $u$ was a part of some \Cref{rdc:leaves}, then it is an inner vertex in $T[\NEOCOL(r)]$ and we have $d_r + 1 = \rho_1 + \EDN(T') = \gamma_c(\NEOCOL(r)) + 1 = \omega(\NEOCOL(r))$.
  Otherwise, $u$ is adjacent only to $r$ in $H$ and we have $\rho_1 = 0$ and $\omega(\NEOCOL(r)) = \EDN(T') = 1$.
  Thus, in any case, we have
  \begin{equation}\label{eq:first-part-weight}
    \omega(V_1) = \rho_1 + \EDN(T') = \rho_1 + 1
  \end{equation}
  Therefore, the total weight of the neo-colonization $\NEOCOL$ is
  \[
    %\EDN(T) = \rho + \EDN(T') = \sum_{i=1}^t \rho_i + \EDN(T') = \sum_{i=1}^t \omega(V_i) = \omega(\NEOCOL).
    \omega(\NEOCOL) =
    \sum_{i=1}^t \omega(V_i) \putabove{=}{(\ref{eq:first-part-weight})}
    \rho_1 + 1 + \sum_{i=2}^t \omega(V_i) \putabove{=}{(\ref{eq:part-weight})}
    1 + \sum_{i=1}^t \rho_i \putabove{=}{(\ref{eq:algo-works})}
    \EDN(T)
  \]
  Thus the total weight is equal to the minimum required number of guards to defend indefinitely.
\end{proof}
A neo-colonization implies the following defending strategy described by Goddart et al.~\cite{eternal-security-in-graphs}.
We will only move guards in the part that was last attacked.
If the attacked $V_i$ induces a clique, move the only guard assigned to it.
Otherwise, there are $\gamma_c(T[V_i]) + 1$ guards on $V_i$.
Keep the vertices of the minimum connected dominating set of $T[V_i]$ always occupied.
We say that the guard which is not placed on the minimum connected dominating set of $T[V_i]$ is the \emph{extra} guard of $V_i$.

In case of an attack, find a path from the extra guard to the attacked vertex and move the guards along this path.
Note that all vertices of the path except the attacked one must be occupied.

We will define the \emph{canonical strategy}, which is a slight modification of the strategy that follows from the nice neo-colonization.
By $h(V_i)$ we denote the vertex closest to $r$ in $V_i$.

\begin{definition}
  Let the \emph{canonical strategy} for tree $T$ rooted in a leaf $r$ be the strategy that follows from the nice neo-colonization $\NEOCOL$ of $T$ with the following modifications.
  \begin{itemize}
    \item When initially placing the guards on $T$, for every $V_i \in \NEOCOL$ which is not a clique, place the extra guard of $V_i$ on $h(V_i)$.
    \item In every part $V_i \in \NEOCOL$ which was not attacked this turn, we move the extra guard to $h(V_i)$.
  \end{itemize}
\end{definition}

Next, we partition the vertices of $T$ into three sets based on their neighborhood in their respective part: $L$ consists of \emph{leaf} vertices, $J$ consists of \emph{joining} vertices, and $I$ consists of \emph{inner vertices}.
\[
\begin{aligned}
  L &= \{ v \mid |\CHLD(v)| = 0 \} \\
  J &= \{ v \mid |\CHLD(v)| = 1 \}  \\
  I &= \{ v \mid |\CHLD(v)| \geq 2 \}
\end{aligned}
\]
%\begin{definition}\label{def:lci}
  %A $v \in T$ is in 
  %\begin{itemize}
    %\item the set of \emph{leaf} vertices $L$ if it has no children which are in the same shrubbery,
    %\item the set of \emph{joining} vertices $J$ if it has one child which is in the same shrubbery,
    %\item the set of \emph{inner} vertices $I$ if it has at least two children which are in the same shrubbery.
  %\end{itemize}
%\end{definition}

For $V_i \in \NEOCOL$, we use $L(V_i)$, $J(V_i)$, and $I(V_i)$ to denote $L \cap V_i$, $J \cap V_i$, and $I \cap V_i$ respectively.
Note that $J$ contains exactly one vertex from each $T[V_i]$ which has the minimum depth.
The following observation notes when the vertices of $L$, $C$, and $I$ are occupied in the canonical strategy.
%We say that a shrubbery was attacked when the attack came to any of its leaves.
\begin{observation}
\begin{itemize}
  \item A leaf vertex is occupied if and only if it was attacked,
  \item a joining vertex is occupied if no leaves in its part were attacked,
  \item an inner vertex is always occupied.
\end{itemize}
\end{observation}

Note that if a part was not attacked, then it moves to a configuration where all its inner vertices and its joining vertex are occupied.
Hence, a part moves its guards only if it is attacked in this or the previous turn.

\subsection{Attacking strategy}

To determine which vertex we want to attack, we will use the notion of a \emph{canonical} number of guards on a vertex or a subtree, based on which vertices would be occupied in the canonical strategy.
By $\CN(v, a)$ we denote the number of guards on $v \in V(T)$ in the canonical strategy after an attack on $a$.
Similarly by $\CN(X, a)$ we denote the total number of guards on $X \subseteq V(T)$ in the canonical strategy after an attack on $a$.
By $\CN_T(v, a)$ we denote $\CN(T(v), a)$.
Note that $\CN_T(v, a) = \CN(v, a) + \sum_{d \in \CHLD(v)} \CN_T(d, a)$.

\begin{lemma}
The following equality holds.
\[
  \CN(v, a) = \begin{cases}
    1 & \text{if } a = v\\
    1 & \text{if } v \in I\\
    1 & \text{if } v \in J \text{ and } a \notin L(\NEOCOL(v)) \\
    0 & otherwise
  \end{cases}
\]
\end{lemma}
\begin{proof}
  We consider the values of $\CN(v, a)$ case by case.
  Any vertex which was attacked must be occupied.  
  In the canonical strategy, all inner vertices are always occupied.
  Exactly one non-inner vertex in each part is occupied -- a leaf if the part was attacked, its joining vertex otherwise.
\end{proof}

Let $C \subseteq V(T)$ be a configuration of guards on $T$.
By $\G(v, C)$ we denote the number of guards on $v$ in $C$ and by $\GT(v,C)$ we denote the number of guards on $T(v)$ in $C$.
Let $\DF(v, a, C) = \G(v,C) - \CN(v, a)$ and $\DFT(v, a, C) = \GT(v,C) - \CN_T(v, a)$.
We say that $\DF$ or $\DFT$ is the \emph{deficit} of vertex $v$ or subtree $T(v)$, respectively.
We say that $v$ is \emph{deficient} if $\DFT(v, a, C) < 0$.
%Note that in our situation, we assume that the defender has less than the required number of guards to win, therefore $r$ is always deficient.

The general idea of the attacking strategy is to find a subtree with less than the necessary number of guards to defend itself.
We will show that the deepest root of such a subtree must be unoccupied, and therefore acts as a bottleneck for guards entering the subtree.

%\begin{theorem}\label{thm:v_trees}
  %Let $C$ be a configuration of guards on $T$ of at most $\EDN(T) - 1$ guards.
  %Then the attacker can win against $C$ in at most ${\rm diam}(T)$ steps, where ${\rm diam}(T)$ is the diameter of the tree.
%\end{theorem}

\attackstrategy

\begin{proof}
  First, we root the tree $T$ in an arbitrary leaf $r$, partition the vertices of the tree into a nice neo-colonization $\NEOCOL = \{V_1, \dots, V_t\}$, and partition the vertices into leaves, joining, and inner, i.e into sets $L, J$, and $I$ respectively.

  Let $v$ be a deficient vertex of maximum depth.
  We will show that such vertex always exists by showing that $r$ is always deficient.
  Consider the canonical strategy on $T$.
  Its number of guards is equal to the weight of neo-colonization $\NEOCOL$, on which the canonical strategy is based. 
  Together with \Cref{lem:neocol-construction}, this implies 
    $\CN(r,a) = \omega(\NEOCOL) = \EDN(T)$
  for any $a \in V(T)$.
  Moreover, $\GT(r, C) \leq \EDN(T) - 1$ and thus $\DFT(r, a, C) = \GT(r, C) - \CN(r, a) < 0$ for any $a \in V(T)$, therefore $r$ is deficient.

  With each attack, we will increase the depth of the deficient vertex of maximum depth.
  Eventually, when $v$ is a leaf of $T$, the fact that $v$ is deficient implies that there was an attack on $v$ which was not defended.

    We will show that there exists a vertex such that we either won by attacking it the previous turn or there is a vertex such that after attacking it, the depth of the lowest deficient vertex increases.

    %Let $d_1, \dots, d_k$ be the children of $v$ and $a$ the previously attacked vertex.
    If this is the first turn of the game, let $a$ be any vertex occupied by a guard.
    Otherwise, let $a$ be the previously attacked vertex.
    Then the following equalities hold.
    \[
      \begin{aligned}
        \DFT(v, a, C) &= \GT(v,C) - \CN_T(v, a) = g(v, C) - \CN(v, a) + \sum_{d \in \CHLD(v)}\!\GT(d, C) - \CN_T(d, C) \\
        &= \DF(v, a, C) + \sum_{d \in \CHLD(v)}\!\DFT(d, a, C)
      \end{aligned}
    \]

    Note that $\DFT(d, a, C) \geq 0$ for every $d \in \CHLD(v)$ by the choice of $v$.
    Furthermore, $\DFT(d, a, C)~\leq~0$ as $\DFT(d, a, C) > 0$ would imply $\DF(v, a, C) < -1$ and that in turn implies $g(v, C) < 0$.
    Therefore, $\DFT(d, a, C) = 0$, from which follows $\DF(v, a, C) = -1$, i.e., $v$ is not occupied and $\CN(v, a) = 1$.

    Now we show that we won the game by the attack on $a$ or $\NEOCOL(v)$ contains a vertex such that when we attack it, the depth of the deepest vertex increases.
    If $v \in L$, then $\CN(v, a) = 1$ and therefore $v = a$ while $v$ is not occupied, thus we won the game by the attack on $a$.
    Thus, let us assume that $v \in I \cup J$.
    
    Let $x$ be the vertex of $\NEOCOL(v)$ on which the canonical strategy would place the extra guard of $\NEOCOL(v)$, i.e., $x \in \NEOCOL(v) \cap (J \cup L)$ and $\CN(x, a) = 1$.
    Such vertex always exists and is uniquely defined -- either $a \in L(\NEOCOL(v))$, in which case $x = a$, or $x$ is the only vertex in $J(\NEOCOL(v))$.
    Also, note that after every attack, the value of $\CN(V_i)$ for every $V_i \in \NEOCOL$ remains unchanged and is equal to $|I(V_i)| + 1$.

    Now we show how to find the vertex which we want to attack.
    For illustration, see \Cref{fig:where-to-attack}.
    First, we will show that there exists $d \in \CHLD(v)$ such that $x \notin T(d)$.
    Suppose that $v \in I$.
    Then $|\CHLD(v)| \geq 2$ and as all subtrees of children of $v$ are vertex disjoint, $x \in T(d')$ for at most one $d' \in \CHLD(v)$.
    Therefore, $d \in \CHLD(v)$ such that $x$ is not in the subtree of $d$ exists.
    % and therefore for every $l \in T(d_i) \cap L \cap \NEOCOL(v)$ it holds $\CN(l, a) = 0$.
    Otherwise, $v \in J$ and therefore $g(v) = 0$ and $\CN(v) = 1$ which implies $x = v$.
    Thus the only child of $v$ does not have $x$ in its subtree.

    \begin{figure}[]
        \centering
        \includegraphics[scale=1.2,page=1]{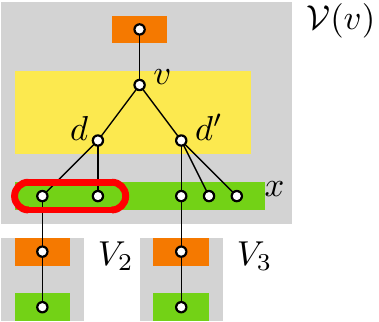}
        \hspace{1cm}
        \includegraphics[scale=1.2,page=2]{images/where-to-attack.pdf}
        \caption{The grey rectangles circumscribe individual parts of the neo-colonization. Orange marks vertices in $J$, yellow marks vertices in $I$ and green marks vertices in $L$. The red border shows the set of vertices from which we can choose the one to attack.}
        \label{fig:where-to-attack}
    \end{figure}

    Now suppose we choose any $b \in T(d) \cap L(\NEOCOL(v))$ for the attack.
    We show that $\CN_T(d, b) = \CN_T(d, a) + 1$.
    Let $V_d = V(T(d))$.
    Note that for every $V_i \in \NEOCOL$ such that $V_i \subseteq V_d$ it holds $\CN(V_i, a) = |I(V_i)| + 1 = \CN(V_i, b)$.
    It also holds $V_d \cap J(\NEOCOL(v)) = \emptyset$ and for every $\ell \in V_d \cap L(\NEOCOL(v))$ such that $\ell \neq b$ it holds $\CN(\ell, a) = \CN(\ell, b) = 0$ while $\CN(b, a) = 0$ and $\CN(b, b) = 1$.
     Let $\NEOCOL' = \{ V_i \mid V_i \subseteq V_d \}$, then
    \[
      \begin{aligned}
        \CN_T(d, b) &= \sum_{V_i \in \NEOCOL'} \CN(V_i, b) + \CN(\NEOCOL(v) \cap V_d, b) \\
        &= \sum_{V_i \in \NEOCOL'} \CN(V_i, b) + |I(\NEOCOL(v)) \cap V_d| + \underbrace{|\{b\} \cap L(\NEOCOL(v)) \cap V_d|}_{= 1}\\
        &= \sum_{V_i \in \NEOCOL'} \CN(V_i, a) + |I(\NEOCOL(v)) \cap V_d| + \underbrace{|\{a\} \cap L(\NEOCOL(v)) \cap V_d|}_{= 0} + 1\\
        &= \sum_{V_i \in \NEOCOL'} \CN(V_i, a) + \CN(\NEOCOL(v) \cap V_d, a) = \CN_T(d, a) + 1. \\
      \end{aligned}
    \]

    Also, the number of guards on $T(d)$ remains unchanged after the attack, as the only vertex from which a guard may enter $T(d)$ is $v$ and it is unoccupied.
    Let $C'$ be the configuration of guards by which the defender responds to the attack on $b$.
    It holds 
    \[
      \DFT(d, b, C') = g_T(d, C') - \CN_T(d, b) = g_T(d, C) - \CN_T(d, a) - 1 = \DFT(d, a, C) - 1 = -1
    \]
    thus $d$ is now a deficient vertex of greater depth than $v$.

    Therefore, after at most ${\rm diam}(T)$ attacks, the lowest deficient vertex must be a leaf, at which point the game is won by the attacker.
\end{proof}

This implies the following result.
\begin{corollary}\label{cor:result}
  Let $T$ be a tree, then $\MINAT(T) \leq n$.
\end{corollary}

\section{Open problems}
Does there exist a graph $G$ on $n$ vertices that can not be defeated in $n$ attacks when the defender has at most $\EDN(G) - 1$ guards?

\bibliographystyle{plainurl}
\bibliography{main}

%\iffullpaper
%\else
%\newpage
%\appendix
%\section*{Appendix}
%\appendixProofText
%\fi

\end{document}